\documentclass[12pt]{amsart}
\usepackage{amssymb}
\usepackage{graphicx}

\usepackage{xcolor} 

\usepackage{fullpage} 

\definecolor{green}{rgb}{0,0.8,0} 

\newtheorem{theorem}{Theorem}[section]
\newtheorem{corollary}[theorem]{Corollary}
\newtheorem{lemma}[theorem]{Lemma}
\newtheorem{proposition}[theorem]{Proposition}
\theoremstyle{definition}
\newtheorem{definition}[theorem]{Definition}

\theoremstyle{remark}
\newtheorem{remark}[theorem]{Remark}
\numberwithin{equation}{section}

\newcommand{\abs}[1]{\vert#1\vert}

\newcommand{\set}[1]{\{#1\}}

\newcommand{\ud}{\mathrm{d}}
\newcommand{\rd}{\partial}

\newcommand{\alp}{\alpha}
\newcommand{\bt}{\beta}
\newcommand{\gmm}{\gamma}
\newcommand{\Gmm}{\Gamma}
\newcommand{\dlt}{\delta}

\newcommand{\eps}{\epsilon}
\newcommand{\ep}{\epsilon}

\newcommand{\Omg}{\Omega}

\newcommand{\bbR}{\mathbb R}

\newcommand{\calE}{\mathcal E}

\newcommand{\calH}{\mathcal H}
\newcommand{\calI}{\mathcal I}

\newcommand{\calQ}{\mathcal Q}

\newcommand{\calS}{\mathcal S}

\newcommand{\PD}{\calQ}

\newcommand{\uC}{\underline{C}}

\newcommand{\pfstep}[1]{\vspace{.5em} \noindent {\it #1.}}

\newcommand{\rst}{\upharpoonright}
\newcommand{\f}{\frac}
\newcommand{\ls}{\lesssim}
\newcommand{\de}{\delta}

\newcommand{\dc}{\gmm}

\begin{document}

\title[]{Dynamical black holes with prescribed masses in spherical symmetry}
\author{Jonathan Luk}
\address{Department of Mathematics, Stanford University, Palo Alto, CA, USA}
\email{jluk@stanford.edu}

\author{Sung-Jin Oh}
\address{Korea Institute for Advanced Study, Seoul, Korea}
\email{sjoh@kias.re.kr}

\author{Shiwu Yang}
\address{Peking University, Beijing, China}
\email{shiwuyang@math.pku.edu.cn}


\begin{abstract}
We review our recent work on a construction of spherically symmetric global solutions to the Einstein--scalar field system with large bounded variation norms and large Bondi masses. We show that similar ideas, together with Christodoulou's short pulse method, allow us to prove the following result: Given $M_i \geq M_f>0$ and $\epsilon>0$, there exists a spherically symmetric (black hole) solution to the Einstein--scalar field system such that up to an error of size $\epsilon$, the initial Bondi mass is $M_i$ and the final Bondi mass is $M_f$. Moreover, if one assumes a continuity property of the final Bondi mass (which in principle follows from known techniques in the literature), then for $M_i>M_f>0$, the above result holds without an $\epsilon$-error.
\end{abstract}
\maketitle

\section{Introduction}

We study the Einstein--scalar field system for a Lorentzian manifold $(\mathcal M, g)$ and a real-valued function $\phi:\mathcal M\to \mathbb R$:
\begin{equation}\label{ESS}
\begin{cases}
Ric_{\mu\nu}-\f 12 R g_{\mu\nu}=2T_{\mu \nu}\\
T_{\mu\nu}=\rd_\mu\phi\rd_\nu\phi-\f 12 g_{\mu\nu}(g^{-1})^{\alp\bt}\rd_\alp\phi\rd_\bt\phi\\
\Box_g \phi=0
\end{cases}
\end{equation}
in $(3+1)$ dimensions with \emph{spherically symmetric data}. Here, $\Box_g$ is the Laplace--Beltrami operator associated to the metric $g$.

In a recent paper \cite{LOY}, we introduced a method of constructing global solutions (in the sense of causal geodesic completeness) with arbitrarily large bounded variation norms and Bondi masses. This is in contrast to the global solutions constructed by Christodoulou \cite{Christodoulou:1986, Christodoulou:1993bt} which have small and decaying initial data. The key innovation is the introduction of a pair of \emph{dimensional} norms which are consistent with \emph{non-decaying} initial data. On the one hand, the smallness of these norms allows us to prove a global existence result and on the other hand, since the data are not required to decay, one can construct data which are \emph{large} in any integrated norms. We refer the readers to Theorem~\ref{thm:main.finite} below for a precise statement.

On a different note, it is well-known that \eqref{ESS} admits spherically symmetric solutions which have a black hole region and are future causally geodesically incomplete \cite{Christodoulou:1991}. In this note, we are particularly interested in constructing black hole spacetimes with special properties using ideas in \cite{LOY}.

More precisely, we show that the estimates proved in \cite{LOY}, together with the short pulse method of Christodoulou \cite{Chr}, allows us to construct black hole spacetimes whose initial and final Bondi masses can be prescribed up to an arbitrarily small error (cf. Theorem~\ref{thm:main}). The key point is that the estimates in \cite{LOY} allow us to construct initial data such that despite black hole formation, the radiation of mass along future null infinity can be well-controlled. 

To further discuss our results, we recall the reduction of \eqref{ESS} in spherical symmetry. It is well-known that in spherical symmetry we can introduce null coordinates $(u,v)$ such that the metric $g$ takes the form
$$g=-\Omg^2 \ud u \cdot \ud v+r^2 \ud \sigma_{\mathbb S^2},$$
where $\ud \sigma_{\mathbb S^2}$ is the standard metric on the unit round sphere and $r$ is the area-radius of the orbit of the symmetry group $SO(3)$. We normalize the coordinates so that $u = v$ on the axis of symmetry $\Gmm = \set{r = 0}$. 
Define the Hawking mass $m$ by the relation
\begin{equation}\label{def:mass}
\Omg^{2} := - \frac{4 \rd_{u} r \rd_{v} r}{1 - \frac{2m}{r}}.
\end{equation}
The solutions that we will construct have the following Penrose diagrams\footnote{In our theorem, we only control part of the maximal globally hyperbolic developments. The fact that the solutions indeed have the Penrose diagrams as claimed is a result of Christodoulou \cite{Christodoulou:1991}.}:

\begin{figure}[h] \label{fig:BH}
\begin{center}
\def\svgwidth{130px}
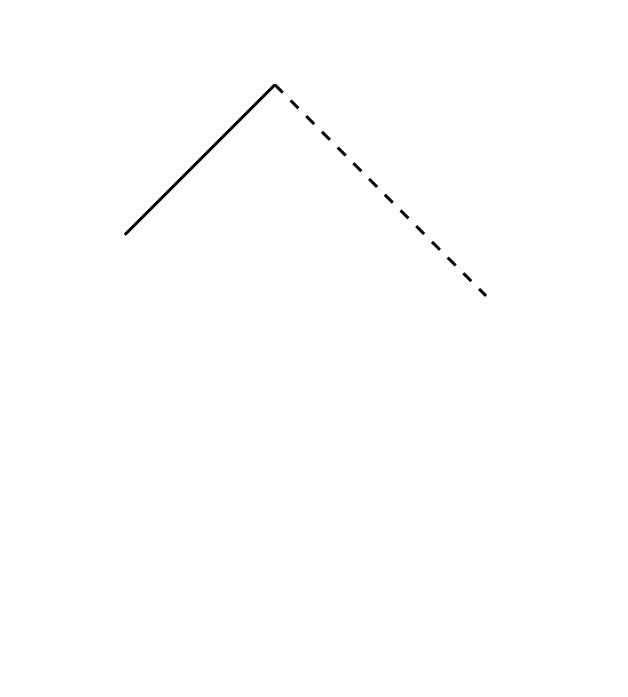 
\caption{}
\end{center}
\end{figure}

In particular, the solutions we construct arise from characteristic initial data given on an outgoing null cone. According to the Penrose diagram above, one can define null infinity as an idealized boundary which is in $1$-to-$1$ correspondence to outgoing radial null curves such that $r\to\infty$. The event horizon is then define to be the future boundary of the past of null infinity. In an appropriate double null $(u,v)$ coordinate system, we define 
\begin{itemize}
\item $u=u_0$ corresponds to the initial outgoing null cone,
\item $u=u_{\mathcal E\mathcal H}$ is the event horizon,
\item $v=\infty$ to (formally) be future null infinity. 
\end{itemize}
We can now defined the notions of initial and final Bondi masses as appropriate limits of the Hawking mass (cf. \eqref{def:mass}).
\begin{definition}
Let the \emph{Bondi mass} be a function of $u$ along null infinity:
$$m_{Bondi}(u):=\lim_{v\to \infty} m(u,v).$$
Define also the \emph{initial Bondi mass} $m_i$ and \emph{final Bondi mass} $m_f$ by 
$$m_i:=\lim_{u\to u_0^+} m_{Bondi}(u),\quad m_f:=\lim_{u\to u_{\mathcal E\mathcal H}^-} m_{Bondi}(u).$$
\end{definition}
The following is the main result of this paper:
\begin{theorem}\label{thm:main}
Let $M_f$ and $M_i$ be two arbitrary positive real numbers satisfying $M_i\geq M_f > 0$. Then for every $\ep>0$, there exists a spherically symmetric solution to \eqref{ESS} with Penrose diagram given by Figure~\ref{fig:BH} such that
$$|m_i-M_i|\leq \ep,\quad |m_f-M_f|\leq \ep.$$
\end{theorem}

\begin{remark} [Removing the $\ep$-error]\label{rem:no-eps}
Assuming that $m_{f}$ depends continuously on the initial data in an appropriate topology, we may remove the $\eps$ in Theorem~\ref{thm:main} when $M_{i} > M_{f}$ via a soft continuity argument, and establish the following statement:
\begin{quote}
Let $M_f$ and $M_i$ be two arbitrary positive real numbers satisfying $M_i > M_f$. Then there exists a spherically symmetric solution to \eqref{ESS} with Penrose diagram given by Figure~\ref{fig:BH} such that $m_{i} = M_{i}$ and $m_{f} = M_{f}$.
\end{quote}
We will present this soft continuity argument, along with the precise continuity property of $m_{f}$ that we require (Proposition~\ref{prop:mf-cont-dep}), in Section~\ref{subsec:no-eps} below.

The needed continuity property of $m_{f}$ (Proposition~\ref{prop:mf-cont-dep}) should be a consequence of asymptotic stability of the exterior of the Schwarzschild spacetime as a solution to \eqref{ESS} in spherical symmetry. Unfortunately, to our knowledge there is no direct account of this result in the literature, although it can be inferred using techniques in some existing related works \cite{DR, holzegel, LO-ext}. 
To avoid lengthening the paper, we have chosen \underline{not} to pursue the proof of Proposition~\ref{prop:mf-cont-dep} (hence keeping the $\eps$ in the statement our main theorem). 
\end{remark}

\begin{remark}[Monotonicity of Bondi mass and positivity of mass]
It is well-known that 
\begin{enumerate}
\item the Bondi mass is non-increasing along future null infinity,
\item the final Bondi mass $m_f\geq 0$.
\end{enumerate} 
These facts can easily be inferred from the equations in Section~\ref{sec.prelim} and the condition $m=0$ at $r=0$. As a consequence, it is necessary to require $M_i\geq M_f\geq 0$ in the main theorem. Notice that allowing $\ep$-errors, we can indeed cover this whole range of mass parameters\footnote{For $M_i>M_f=0$, this can be achieved by repeated applications of Theorem~\ref{thm:main} for a fixed $M_i$ and a sequence $(M_f)_n\to 0$.}. Without $\ep$-errors (see Remark~\ref{rem:no-eps}), our theorem does not cover the case $M_i>M_f=0$: to our knowledge, it is not known whether there are black hole solutions with $m_f=0$. On the other hand, $M_i>M_f=0$ can easily be achieved (exactly) by dispersive solutions using the scaling transformation $(r,\phi,m)\mapsto (ar, \phi, am)$ for $a>0$.
\end{remark}

\begin{remark}[Spacetimes with a complete regular past]
The theorem above is stated for a characteristic initial value problem starting from an outgoing null cone. In fact, proceeding\footnote{Note that the procedure of passing to a limit to obtain a spacetime all the way up to past null infinity has already been introduced in \cite{Chr}.} as in \cite{LOY}, one can consider a sequence of such initial outgoing null cones and pass to a limit to past null infinity to obtain a black hole spacetime with a complete regular past. This can be viewed as ``prescribing data on past null infinity''. \textbf{In this case, the ``initial mass'' can be understood as the ADM mass of the spacetime.} We omit the details.
\end{remark}

Theorem~\ref{thm:main} gives the following 
\begin{corollary}\label{cor:main}
Let $\mathcal S$ be the set of all spherically symmetric characteristic initial data (in a sense to be made precise in Section~\ref{sec.prelim}) for which the maximal globally hyperbolic development has a trapped surface and a black hole region. Then
$$\inf_{\calS} \f{m_f}{m_i}=0,\quad \sup_{\calS} \f{m_f}{m_i}=1.$$
\end{corollary}

We briefly comment on both of these statements. The infimum is motivated in part by the critical collapse picture. This is a conjectural picture put forth by Choptuik \cite{Choptuik}, whose numerical work suggests that for a generic $1$-parameter family rescaling of initial data connecting black hole solutions to dispersive solutions, the masses of the black holes tend to $0$ with a universal rate at the threshold. If this picture is valid, then there are many $1$-parameter families of data sets parametrized by $\lambda\in (0,\lambda_*)$, such that each data set leads to a black hole with positive final mass and such that $\inf_{\lambda} \f{m_f}{m_i} = \lim_{\lambda\to 0} \f{m_f}{m_i} = 0$. While we are far from proving this, we show at least that if one takes the infimum over \emph{all} initial data sets (as opposed to a $1$-parameter family\footnote{Notice moreover that at the limit the Choptuik solutions conjecturally tend to a naked singularity demonstrating discrete self-similarity. However, if one were to take a limit in our case -- even though our estimates are too weak to justify such a limit -- one would heuristically obtain a dispersive solution constructed in \cite{LOY}.}) which lead to a black hole, then the infimum is indeed $0$. 

On the other hand, it is easy to see that the ($2$-ended) Schwarzschild solution with any positive mass parameter achieves the supremum in Corollary~\ref{cor:main}. Our construction in particular gives examples of black hole solutions arising from $1$-ended initial data which are close to a Schwarzschild solution in a large subset of the spacetime. In this special case, the construction can be viewed as a direct adaptation of the ideas of \cite{Chr} to our setting.

\section{Basic setup and main result in \cite{LOY}}\label{sec.prelim}

In this section, we recall the geometric setup in \cite{LOY}, which we will use for the rest of the paper. As in \cite{LOY}, we will adopt the point of view that under spherical symmetry, the (3+1)-dimensional Einstein--scalar field equation \eqref{ESS} reduces to a (1+1)-dimensional system.

Fix $u_0\in \mathbb R$. Consider the domain in $\mathbb R^{1+1}$
\begin{equation*}
	\PD = \set{(u, v) \in \bbR^{1+1} : u \in [u_0, \infty), \, v \in [u, \infty)},
\end{equation*}
with partial boundary $\set{(u_0,v)\in \PD: v\in [u_0,\infty)}$ and
\begin{equation*}
	\Gmm = \set{(u, u) \in \PD : u \in [u_0, \infty)}.
\end{equation*}
We define causality in $\PD$ with respect to the ambient metric $m = - \ud u \cdot \ud v$ of $\bbR^{1+1}$, and the time orientation in $\PD$ so that $\rd_{u}$ and $\rd_{v}$ are future pointing. $C_{u}$ and $\uC_{v}$ will be used to denote constant $u$ and $v$ curves in $\PD$, respectively. 
Moreover, given $- \infty < u_{0} < u_{1} < \infty$, let
\begin{align*}
	\PD_{[u_{0}, u_{1}]} =& \set{(u, v) \in \PD : u \in [u_{0}, u_{1}]}.
\end{align*}
We endow on $\PD$ a Lorentzian metric
$$g=-\Omg^2 \ud u \cdot \ud v+r^2 \ud \sigma_{\mathbb S^2},$$
where $\ud \sigma_{\mathbb S^2}$, as before, is the standard round metric on $\mathbb S^2$ with radius $1$. We require $\Omg^2>0$ and $r\geq 0$, with equality exactly at $\Gmm$.

In spherical symmetry, \eqref{ESS} reduces to the following system of $(1+1)$-dimensional wave equations for $(r,\phi,\Omg)$
\begin{equation}\label{WW.SS}
\left\{
\begin{aligned}
\rd_u\rd_v r=&-\f {\Omg^2}{4 r}-\f{\rd_u r\rd_v r}{r},\\
\rd_u\rd_v \phi=&-\f{\rd_v r\rd_u\phi}{r}-\f{\rd_u r\rd_v\phi}{r},\\
\rd_u\rd_v\log\Omg=&-\rd_u\phi\rd_v\phi+\f{\Omg^2}{4r^2}+\f{\rd_u r\rd_v r}{r^2},
\end{aligned}
\right.
\end{equation}
coupled with the Raychaudhuri equations
\begin{equation}\label{eqn.Ray}
\left\{
\begin{aligned}
\rd_v\left(\f{\rd_v r}{\Omg^2}\right)=&-\f{r(\rd_v\phi)^2}{\Omg^2},\\
\rd_u\left(\f{\rd_u r}{\Omg^2}\right)=&-\f{r(\rd_u\phi)^2}{\Omg^2}.
\end{aligned}
\right.
\end{equation}

In terms of the Hawking mass (cf. \eqref{def:mass}),
\eqref{ESS} reduces to the following system of equations for $(r,\phi, m)$ in $(1+1)$ dimensions:
\begin{equation} \label{eq:SSESF} 
\left\{
\begin{aligned}
	\rd_{u} \rd_{v} r = & \frac{2m \rd_{u} r \rd_{v}r }{(1-\frac{2m}{r}) r^{2}}, \\
	\rd_{u} \rd_{v} (r \phi) = & \frac{2m \rd_{u} r \rd_{v}r }{(1-\frac{2m}{r}) r^{2}} \phi, \\
	\rd_{u} r \rd_{u} m = & \frac{1}{2} (1-\f{2m}{r}) r^{2} (\rd_{u} \phi)^{2} , \\
	\rd_{v} r \rd_{v} m = & \frac{1}{2} (1-\f{2m}{r}) r^{2} (\rd_{v} \phi)^{2} ,
\end{aligned}
\right.
\end{equation}
As long as $1-\f{2m}r\neq 0$, this is equivalent to the combined system consisting of \eqref{WW.SS} and \eqref{eqn.Ray}. In what follows, it will be convenient to use both formulations of the equations.

We recall the notion of a $C^{k}$ solution to \eqref{eq:SSESF} given in \cite{LOY}.
\begin{definition} \label{def:C1-sol}
Let $- \infty < u_{0} < u_{1} < \infty$. We say that a triple $(r, \phi, m)$ of real-valued functions on $\PD_{[u_{0}, u_{1}]}$ is a \emph{$C^{k}$ solution to \eqref{eq:SSESF}} if it satisfies this system of equations and the following conditions hold:
\begin{enumerate}
\item The following functions are $C^{k}$ in $\PD_{[u_{0}, u_{1}]}$:
\begin{equation*}
	\rd_{u} r, \ \rd_{v} r, \ \phi, \ \rd_{v}(r \phi), \ \rd_{u}(r \phi).
\end{equation*}
\item For $\rd_{v} r$ and $\rd_{u} r$, we have 
\begin{equation*}
	\inf_{\PD_{[u_{0}, u_{1}]}} \rd_{u} r > -\infty, \quad \inf_{\PD_{[u_{0}, u_{1}]}} \rd_{v} r > 0.
\end{equation*}
\item For each point $(a, a) \in \Gmm \cap \PD_{[u_{0}, u_{1}]}$, the following boundary conditions hold:
\begin{align} 
\label{eq:bc4r}
	r (a, a) =& 0, \\
\label{eq:bc4m}
	m(a, a) =& 0.
\end{align}
\end{enumerate}

\end{definition}

The boundary condition \eqref{eq:bc4r} can be combined with the regularity assumption to deduce higher order boundary conditions for $r$ and $r \phi$. More precisely, let $(r, \phi, m)$ be a $C^{k}$ solution on $\PD_{[u_{0}, u_{1}]}$. Since $u = v$ on $\Gmm = \set{ r = 0}$, we have
\begin{equation} \label{eq:bc4high-d}
	(\rd_{v} + \rd_{u})^{\ell} r (a, a) = 0, \quad 
	(\rd_{v} + \rd_{u})^{\ell} (r \phi) (a, a) = 0, 
\end{equation}
for every $\ell = 0, \ldots, k$ and $(a, a) \in \Gmm \cap \PD_{[u_{0}, u_{1}]}$.

Consider the characteristic initial value problem for \eqref{eq:SSESF} with data 
\begin{equation} \label{eq:ini-dvrphi}
	(\rd_v r)^{-1} \rd_{v} (r \phi) \restriction_{C_{u_{0}}} = \Phi, 
\end{equation}
and initial gauge condition\footnote{We call \eqref{eq:ini-dvr} an initial gauge condition since it can be enforced for an arbitrary initial data set by a suitable reparametrization of the coordinate $v$, which is a gauge symmetry of the problem. See Remark~\ref{rem:gauge}.}
\begin{equation} \label{eq:ini-dvr}
	(\rd_v r) \restriction_{C_{u_{0}}} = \frac{1}{2},
\end{equation}
on some outgoing null curve $C_{u_{0}}$. This problem is locally well-posed for $C^{k}$ data $(k \geq 1)$ in the following sense: Given any $C^{k}$ data $\Phi$ with $k \geq 1$, there exists a unique $C^{k}$ solution to \eqref{eq:SSESF} on $\PD_{[u_{0}, u_{1}]}$ for some $u_{1} > u_{0}$, which only depends on $u_{0}$ and the $C^{k}$ norm of $\Phi$. We omit the standard proof.

\begin{remark} \label{rem:gauge}
The system \eqref{eq:SSESF} is invariant under reparametrizations of the form $(u, v) \mapsto ( U(u), V(v) )$; this is the \emph{gauge invariance} of \eqref{eq:SSESF}. Note that we have implicitly fixed a gauge in the setup above, by requiring that $u = v$ on $\Gmm$ and imposing the initial gauge condition \eqref{eq:ini-dvr}.
\end{remark}

We have introduced all the necessary conventions to recall the main result of \cite{LOY}:
\begin{theorem} \label{thm:main.finite}
Consider the characteristic initial value problem from an outgoing curve $C_{u_0}$ with $v\geq u_0$, $\rd_v r \rst_{C_{u_0}} = \frac{1}{2}$ and $r(u_0,u_0)=m(u_0,u_0)=0$.
Suppose the data on the initial curve $C_{u_0}$ are given by
\[
2 \rd_{v}(r \phi)(u_0, v)=\Phi(v),
\]
where $\Phi:[u_0,\infty)\to\mathbb R$ is a $C^1$ function satisfying the following conditions for some $\dc>0$:
\begin{equation}
 \label{eq:IDcond}
\int_{u}^{v}|\Phi(v')|\, \ud v'\leq \eta(v-u)^{1-\dc}, \quad |\Phi'(v)| \leq \eta, \quad |\Phi(v)| \leq \eta, \quad \forall v\geq u\geq u_0.
\end{equation}
Then there exists $\eta>0$ depending only on $\dc$ such that the the unique solution to \eqref{eq:SSESF} arising from the given data is future causally geodesically complete. Moreover, the solution satisfies the following uniform a priori estimates\footnote{The estimates $\rd_v r\leq \f 12$ and $\f{2m}{r}\geq 0$ are not explicitly stated in \cite{LOY}. However, these estimates are almost trivial. $\rd_v r\leq \f 12$ holds since by \eqref{eq:main:geo}, the right hand side of the $\rd_u\rd_v r$ equation in \eqref{eq:SSESF} is negative, and $\rd_v r=\f 12$ initially. $\f{2m}{r}\geq 0$ simply follows from positivity of $m$ and $r$.}:
\begin{equation} \label{eq:main:geo}
	\f 12\geq \rd_v r>\f 13,\quad -\f 16>\rd_u r >-\f 23,\quad 0 \leq \f{2m}{r}<\f 12, \\
\end{equation}
and
\begin{equation} \label{eq:main:apriori}
\begin{gathered}
	\abs{\phi} \leq C \eta \min \set{1, r^{-\gamma}}, \quad
	\abs{\rd_{v} (r \phi)} \leq C ( \abs{\Phi(v)} + \eta \min \set{1, r^{-\gamma}}), \quad
	\abs{\rd_{u} (r \phi)} \leq C \eta, \\
\abs{\rd_{v}^{2} (r \phi)} + \abs{\rd_{v}^{2} r} + \abs{\rd_{u}^{2} (r \phi)} + \abs{\rd_{u}^{2} r} \leq C \eta
\end{gathered}
\end{equation}
for some constant $C>0$ depending only on $\dc$.
\end{theorem}

\begin{remark} 
The first two estimates in \eqref{eq:IDcond} correspond to initial smallness of the pair of \emph{dimensional} norms mentioned in the introduction.
\end{remark}

\begin{remark}
We will \underline{not} be using Theorem~\ref{thm:main.finite} to construct complete solutions. Instead, we will use Theorem~\ref{thm:main.finite} together finite speed of propagation to obtain a regular solution in a large region of spacetime with good estimates.
\end{remark}

\section{Proof of Theorem~\ref{thm:main}}
\subsection{Choice of initial data and remarks on the smallness and largeness parameters} \label{subsec:id}

Let $M_i\geq M_f>0$ and $\ep>0$ as in Theorem~\ref{thm:main} are given. Without loss of generality, we can assume that $\ep$ is small with respect to $M_i$ and $M_f$. In the argument below, we fix $u_{0} = 0$.

We introduce the large parameters $1\ll v_* \ll v_0$ and the small parameter $0< \eta_{0}, \,\de\ll 1$, all of which are to be chosen later. 

To specify the initial data, we first define a $C^\infty_c([0,1])$ function $\chi$ satisfying
\begin{equation}\label{chi.prop}
\int_0^1 \chi(s)\, \ud s=0,\quad \int_0^1 \chi^2(s)\, \ud s=1,\quad -2\leq  \chi\leq 2.
\end{equation}
Also, let $\eta \in [0, \eta_{0}]$.

\begin{definition}\label{def:data}
We define the initial data as follows:
\begin{itemize}
\item For $v\in [0, v_0]$,
$$\Phi(v)=\sum_{i=1}^{\lfloor \f{M_i-M_f}{\eta_{0}^2} \rfloor } 2 \eta \chi(v-v_*-i^2).$$
\item For $v\in [v_0, v_0+\de]$,
$$\Phi(v)=2 M_f^{\f 12} \de^{-\f 12} \chi(\f{v-v_0}{\delta}).$$
\item For $v\geq v_0+\delta$, let $\Phi(v)\equiv 0$.
\end{itemize}
\end{definition}

It will be convenient to denote
$$u_0' = v_*+\left(\lfloor \f{M_i-M_f}{\eta_{0}^{2}} \rfloor\right)^2+1.$$
Let us note that by definition,
\begin{equation}\label{Phi.vanishing}
\Phi(v)=0\quad\mbox{for } u_0' \leq v\leq v_0.
\end{equation}

We remark on the choice of the constants. 
We will first choose $v_*$ large, then $\eta_{0}$ small, then we define
$$v_0=2\left(v_*+\left(\lfloor \f{M_i-M_f}{\eta_{0}^{2}} \rfloor\right)^2+1\right).$$
In particular, we have the useful relation $v_0-\left(v_*+\left(\lfloor \f{M_i-M_f}{\eta_{0}^{2}} \rfloor\right)^2+1\right) \gtrsim v_0$. Finally, we choose $\de$ to be small. Importantly, all constants can be chosen depending on the mass parameters $M_i$ and $M_f$; $\de$ can in addition be dependent on the other parameters.

In view of the above discussions, we will use the following convention:
$\ls$ will be used such that $A\ls B$ means $A\leq CB$ for some universal constant $C>0$. $\ls_{M_f, \, M_f^{-1}}$, $\ls_{M_i}$, etc. will mean that the implicit constants may depend on the mass parameters $M_i$, $M_f$. Finally, $\ls_{v_0}$, $\ls_{M_f, v_0}$, etc. will mean that the implicit constants may, in addition, depend on $v_0$. The last convention is most convenient when the bound is given in terms of $\de$, which can be chosen depending on $v_0$.

We will also use the big-$O$ notation with similar conventions.

\subsection{Computation of the initial mass}
The first order of business is to compute the initial mass for the data given above. This will be given in Proposition~\ref{prop:initial.mass}. First, we need the following lemma.

\begin{lemma}\label{lem:data}
$\rd_v\phi (0, \cdot)$ and $\phi (0, \cdot)$ are both supported in $\left(\displaystyle\bigcup_{i=1}^{\lfloor \f{M_i-M_f}{\eta_{0}^2} \rfloor}[v_*+i^2, v_*+i^2+1]\right)\cup [v_0, v_0+\de]$ and satisfy the following estimates:
$$\left|r|\rd_v\phi|(0, v)-\f 12 \Phi(v) \right|\ls \eta v_*^{-1},\quad |\phi|(0,v)\ls \eta v_*^{-1}. $$
\end{lemma}
\begin{proof}
The support properties follow from Definition~\ref{def:data} and that $\chi$ has mean $0$ (cf. \eqref{chi.prop}). Integrating $\rd_v(r\phi)=\f 12\Phi$ and using the support properties, it also follows that $r|\phi|\ls \eta$ uniformly. (Here, we have used that while $\Phi$ is $O_{M_f}(\de^{-\f 12})$ large, by choosing $\de$ sufficiently small, its integral is bounded $\ls \de^{\f 12} M_f^{\f 12}\ls \eta$.) The estimate for $|\phi|$ then follows from the support properties. Finally, 
$$\left|r|\rd_v\phi|(0, v)-\f 12 \Phi(v) \right|\leq (\rd_v r) |\phi|(0, v)= \f 12|\phi|(0, v)$$
and thus the desired estimate follows from the already proven bound for $|\phi|$.
\end{proof}

\begin{proposition}\label{prop:initial.mass}
The initial mass $m_i:=m(0,\infty)$ satisfies
\begin{equation*}
	m_{i} = \left( \frac{\eta}{\eta_{0}}\right)^{2} (M_{i} - M_{f}) + M_{f} + O_{M_i}(\eta^2+ \eta^{2} v_*^{-\f 12}+M_f v_0^{-1}).
\end{equation*}
\end{proposition}
\begin{proof}
By \eqref{eq:SSESF}, the mass on the $\{u=0\}$ hypersurface can be computed as follows:
\begin{equation}\label{m.data.formula}
 m(0, v) =\f 12 e^{-\int_{0}^v r\f{(\rd_v\phi)^2}{\rd_v r}(0,v') \, \ud v'}\int_{0}^v r^2\f{(\rd_v\phi)^2}{\rd_v r}(0,v') e^{\int_{0}^{v'} r\f{(\rd_v\phi)^2}{\rd_v r}(0,v'') \, \ud v''} \, \ud v' .
\end{equation}
Note that by $\rd_v r=\f 12$, \eqref{eq:main:geo} and Lemma~\ref{lem:data}, the integral that appears in the exponentials can be controlled as follows:
\begin{equation}\label{exp.est}
\int_{0}^{\infty} r\f{(\rd_v\phi)^2}{\rd_v r}(0,v') \, \ud v'\ls \left(\eta^{2} \sum_{i=1}^{\lfloor \f{M_i-M_f}{\eta_{0}^2} \rfloor}(v_*+i^2)^{-1}\right)+M_f v_0^{-1}\ls \eta^{2} v_*^{-\f 12}+M_f v_0^{-1}.
\end{equation}
where we have used $\sum_{i=1}^\infty\f{1}{a^2+i^2}\ls \f{1}{a}$.
Hence, using $\rd_v r=\f 12$, \eqref{chi.prop}, Definition~\ref{def:data} and Lemma~\ref{lem:data},
\begin{equation*}
\begin{split}
m(0, \infty) =&(1+O(\eta^{2} v_*^{-\f 12}+M_f v_0^{-1}))\left(\int_{0}^\infty r^2(\rd_v\phi)^2(0, v')  \, \ud v'\right) \\
=&\int_{0}^\infty \left(\f{\Phi}{2}\right)^2(0,v')  \, \ud v'+O_{M_i}(\eta^{2} v_*^{-\f 12}+M_f v_0^{-1})\\
=& \sum_{i=1}^{\lfloor \f{M_i-M_f}{\eta_{0}^2} \rfloor} \eta^2 \int_{v_*+i^2}^{v_*+i^2+1} \chi^2(v'-v_*-i^2) \, \ud v' + M_f \de^{-1} \int_{v_0}^{v_0+\delta} \chi^2(\f{v'-v_0}{\de}) \, \ud v' \\
&+O_{M_i}(\eta^{2} v_*^{-\f 12}+M_f v_0^{-1})\\
=&\eta^2 \cdot \lfloor \f{M_i-M_f}{\eta_{0}^2} \rfloor +M_f +O_{M_i}(\eta^{2} v_*^{-\f 12}+M_f v_0^{-1})\\
=&\left(\frac{\eta}{\eta_{0}}\right)^{2} (M_{i} - M_{f}) + M_{f} +O_{M_i}(\eta^2+ \eta^{2} v_*^{-\f 12}+M_f v_0^{-1}). \qedhere
\end{split}
\end{equation*}
\end{proof}

\subsection{The region before the short pulse: an application of Theorem~\ref{thm:main.finite}}
A immediate corollary of Theorem~\ref{thm:main.finite} and finite speed of propagation is the following:
\begin{proposition}
For $\eta>0$ sufficiently small, in the region $0\leq u\leq v\leq v_0$, the solution is regular and satisfy the estimates in Theorem~\ref{thm:main.finite} for $\gamma=\f 12$.
\end{proposition}
\begin{proof}
This is an easy verification.
\end{proof}
\begin{remark}
Notice that we \underline{cannot} apply Theorem~\ref{thm:main.finite} for the whole initial data on $[0,\infty)$ since the data are large on $[v_0, v_0+\de]$. In fact, if we were to apply Theorem~\ref{thm:main.finite}, then we would have obtained a geodesically complete spacetime, which is not what we are aiming at.
\end{remark}

With the estimates in Theorem~\ref{thm:main.finite}, we in fact can derive more precise information about the solution, namely, that $\rd_v(r\phi)$ is essentially given by $\Phi$. Here, we need a slightly stronger estimate than in \cite{LOY}.
\begin{proposition}\label{imp.rdvphi}
In the set $\{(u,v): 0\leq u\leq v\leq v_0\}$, we have
$$r|\rd_v\phi|(u,v)\ls_{M_i} |\Phi|(v)+r_+^{-1}\left|\int_u^v \Phi(v')\, \ud v'\right|+\eta r_+^{-1}(u,v)\log (1+r_+(u,v)),$$
where $r_{+} = \max \set{1, r}$.
\end{proposition}
\begin{proof}
We will also take $(u,v)\in \{(u,v): 0 \leq u\leq v\leq v_0\}$ in this proof. By \eqref{eq:SSESF}, we have the integral formula
\begin{align*}
	\rd_{v} (r \phi)(u, v) 
	= \rd_{v} (r \phi)(0, v)+ \int_{u_{0}}^{u} \frac{2 m \rd_u r \rd_{v} r}{(1-\f{2m}{r}) r^{2}} \phi (u', v) \, \ud u'.
\end{align*}
We will estimate $\rd_{v} (r \phi)(u, v) -\rd_{v} (r \phi)(0, v)$ in two ways. Using \eqref{eq:main:geo} (for $(1-\f{2m}{r})^{-1}$ and $\rd_{v} r$), \eqref{eq:main:apriori} (for $\phi$) and Proposition~\ref{prop:initial.mass} (for $m_i\ls M_i$),
\begin{equation}\label{imp.rdvrphi.1}
\left|\rd_{v} (r \phi)(u, v) -\rd_{v} (r \phi)(0, v) \right|\leq \f{2 m_i}{r(u,v)} \left(\sup_{u'\in [0, u]}|\phi|(u',v)\right)\ls_{M_i} \eta r^{-1}(u,v).
\end{equation}
On the other hand, by \eqref{eq:SSESF} and \eqref{eq:main:geo} (for $\rd_v r$), $\left|\int_{u_{0}}^{u} \frac{2 m \rd_u r \rd_{v} r}{(1-\f{2m}{r}) r^{2}} (u', v) \, \ud u'\right|= \left| \rd_{v} r(0, v) - \rd_{v} r (u, v) \right|\leq \frac{1}{6}$. This implies, using \eqref{eq:main:apriori},
\begin{equation}\label{imp.rdvrphi.2}
\left|\rd_{v} (r \phi)(u, v) -\rd_{v} (r \phi)(0, v) \right|\leq \frac{1}{6}\left(\sup_{u'\in [0, u]}|\phi|(u',v)\right)\ls \eta.
\end{equation}
Combining \eqref{imp.rdvrphi.1}, \eqref{imp.rdvrphi.2} and using $\rd_{v} (r \phi)(0, v)=\f{\Phi}{2}(v)$ yield
\begin{equation}\label{imp.rdvrphi.3}
\left|\rd_v(r\phi)(u,v)-\f{\Phi}{2}(v)\right|\ls_{M_i} \eta r_+^{-1}.
\end{equation}
Integrating this in $v$ and using \eqref{eq:main:geo} for $\rd_v r$, we have
\begin{equation*}
r|\phi|(u,v)\ls_{M_i} \left|\int_u^v \Phi(v')\, dv'\right|+ \eta\log (1+r_+(u,v)).
\end{equation*}
On the other hand $|\phi|(u,v)\ls \eta$ by \eqref{eq:main:apriori}. Hence,
\begin{equation}\label{imp.rdvrphi.4}
|\phi|(u,v)\ls_{M_i} r_+^{-1}\left|\int_u^v \Phi(v')\, \ud v'\right|+ \eta r_+^{-1}(u,v)\log (1+r_+(u,v)).
\end{equation}
Since $r\rd_v \phi=\rd_v(r\phi)-(\rd_v r)\phi$, the conclusion follows from $\rd_v r\leq \f 12$ (in \eqref{eq:main:geo}), \eqref{imp.rdvrphi.3} and \eqref{imp.rdvrphi.4}.
\end{proof}

As a consequence, we can estimate the mass on \emph{part}\footnote{Methods in Proposition~\ref{prop:initial.mass} show that $m(0, v_0)\approx M_i$. Hence this estimate does not hold in general on the whole hypersurface.} of the hypersurface $\{v=v_0\}$:
\begin{proposition}\label{small.mass.first.region}
In the set $\left\{(u,v): u_0' \leq u\leq v\leq v_0\right\}$, we have 
$$m(u,v)\ls_{M_i} \eta^2.$$
In particular, this holds on $\left\{(u,v_0): u_0' \leq u\leq v_0\right\}$.
\end{proposition}
\begin{proof}
By \eqref{Phi.vanishing}, for any $(u,v)$ in the given set, $\Phi(v)=0$ and $\int_u^v \Phi(v')\, \ud v'=0$. Hence, using appropriate variants of \eqref{m.data.formula} and \eqref{exp.est}, the estimate \eqref{eq:main:geo} (for $\rd_v r$) and Proposition~\ref{imp.rdvphi}, we have
$$m(u,v)\ls \eta^2 \int_u^v r_+^{-2}(u,v)\log^2 (1+r_+(u,v)) \ud v' \ls \eta^2 \int_0^\infty r_+^{-2}\log^2 (1+r_+) \ud r\ls \eta^2.$$
\end{proof}

\subsection{The short pulse region}

In this section, we recall Christodoulou's short pulse method \cite{Chr} for constructing spacetimes exhibiting trapped surface formation, appropriately adapted to our (much simpler) setting in spherical symmetry.

\begin{proposition}[Short pulse method]\label{prop:shortpulse}
Let 
\begin{equation}\label{u1.def}
u_1=v_0-\f{M_f}{64}.
\end{equation}

Then, for $v_0$ sufficiently large and $\de>0$ sufficiently small (depending on $v_0$),
\begin{itemize}
\item The $2$-sphere given by $(u,v)=(u_1,v_0+\de)$ is trapped, i.e., $(\rd_u r)(u_1, v_0+\de)<0$, $(\rd_v r)(u_1, v_0+\de)<0$.
\item The following estimates hold in the set $\{(u,v):u\in [0, u_1], \, v\in [v_0, v_0+\de]\}$:
\begin{align}
\label{r.est}
\f {M_f}{448} \leq \f 17 (v_0-u) &\leq r(u,v)\leq (v_0-u),\\
\label{duphi.shortpulse}
\left|r\rd_u\phi(u,v)- r\rd_u\phi(u,v_0)\right|&\ls_{M_f, \,M_f^{-1}} \de^{\f 12},\\
\label{dvphi.shortpulse}
\left|r\rd_v\phi(u,v)-r\rd_v\phi(0,v)\right|&\ls_{M_f, \, M_f^{-1}} v_0^2\left(\eta+\de^{\f 12}\right),\\
\label{phi.shortpulse}
|\phi(u,v)|&\ls \eta r_+^{-\f 12}+O_{M_f,\, M_f^{-1}}(\de^{\f 12}),\\
\label{dvrOmg.bound}
\left|\f{\rd_v r}{\Omg^2}(u,v)\right|&\ls 1.
\end{align}
\item If $u\in [0, u_0']$, we have, in addition, that
\begin{equation}\label{dvr.imp.bound}
\left|(\rd_v r)(u,v)-\f 12\right|\ls_{M_i} v_0^{-1}. 
\end{equation}
\end{itemize}
\end{proposition}
\begin{proof}
\pfstep{Step~1: Estimates on (initial) $\{u=0\}$ hypersurface} Since $\phi(0,v_0)=0$ (Lemma~\ref{lem:data}), we have $|\phi|(0,v)\leq r^{-1}(0,v) \cdot \de \cdot (M_f^{\f 12} \de^{-\f 12})\ls v_0^{-1} \de^{\f 12} M_f^{\f 12}$ for $v\in [v_0, v_0+\de]$. Since on $\{u=0\}$, $\rd_v r=\f 12$ and $2\rd_v(r\phi)(0,v)=\Phi(v)$, we have, for $v\in [v_0, v_0+\de]$,
$$r\rd_v\phi(0,v)=\rd_v(r\phi)(0,v)-\f 12 \phi(0,v)=\f{\Phi(v)}{2}+O(v_0^{-1} \de^{\f 12} M_f^{\f 12}).$$
The Raychaudhuri equation \eqref{eqn.Ray}, $\rd_v r(0, v)=\f 12$ and $r (0, v) = \frac{1}{2} v \geq \frac{1}{2} v_{0}$ imply, for $v\in [v_0, v_0+\de]$,
\begin{equation}\label{dvlogOmg.data.shortpulse}
\begin{split}
|\rd_v\log\Omg^2|(0,v)=&2r(\rd_v\phi)^2(0,v)\leq v_0^{-1} |\Phi(v)|^2+O(v_0^{-2}|\Phi(v)|\de^{\f 12} M_f^{\f 12}+v_0^{-3} \de M_f)\\
\leq & 4 v_0^{-1} M_f \de^{-1} + O_{M_f}(v_0^{-2}).
\end{split}
\end{equation}

\pfstep{Step~2: Bootstrap argument} Assume as bootstrap assumptions that for every $(u,v)\in [0, u_1]\times [v_0, v_0+\de]$,
\begin{equation}\label{BA:shortpulse}
\f{1}{18}\leq \Omg^2(u,v) \leq \f{16}{3},\quad r|\rd_v\phi|(u,v)\leq 4M_f^{\f 12}\de^{-\f 12}.
\end{equation}
\textbf{Unless otherwise specified, we will take $(u,v)\in [0, u_1]\times [v_0, v_0+\de]$ in the rest of the proof.}

\pfstep{Step~2(a): Estimates for $\rd_v r$ and $r$ (Proof of \eqref{r.est})} Since by \eqref{WW.SS} $\rd_u(\rd_v r^2)=-\f{\Omg^2}{2}$, by \eqref{BA:shortpulse}, $|\rd_u(\rd_v r^2)|(u,v)\leq \f{8}{3}$. Moreover, since initially $\rd_v r^2(0,v)=2(r \rd_v r)(0,v)=\f{v}{2}$, we have 
\begin{equation}\label{rdvr2}
\left|(\rd_v r^2)(u,v)- \f{v}{2}\right| \leq \f{8 u_1}{3}.
\end{equation}
This implies $\left|r^2(u,v)-r^2(u,v_0)\right|\leq \f{\de v}{2}+\f{8 \de u_1}{3}$. Together with the bound \eqref{eq:main:geo} for $\rd_{v} r$ in $\set{(u, v) : 0 \leq u \leq v \leq v_{0}}$, we have
$$\f 1{9} (v_0-u)^2-\f{\de v_0}{2}-\f{8 \de u_1}{3} \leq r^2(u,v)\leq \f 1 4 (v_0-u)^2+\f{\de v_0}{2}+\f{8 \de u_1}{3}.$$
Choosing $\de$ sufficiently small depending on $v_0$ (note also that $u_1 < v_0$), we obtain \eqref{r.est} (where the left-most inequality in \eqref{r.est} is achieved using \eqref{u1.def}).

Substituting \eqref{r.est} into \eqref{rdvr2}, using $v,\, u_1 \ls v_0$ and $(v_0-u)\gtrsim M_f$ (which follows from \eqref{u1.def}), yields
\footnote{Note that $\rd_v r$ can be negative: it is indeed negative in part of the region as we show in Step~3.}
\begin{equation}\label{dvr.shortpulse}
|\rd_v r|(u,v) \ls \f{v_0}{M_f}.
\end{equation}
\pfstep{Step~2(b): Estimates for $\rd_u r$} Using $\rd_v(\rd_u r^2)=-\f{\Omg^2}{2}$, \eqref{eq:main:geo} (for $\rd_u r$ when $v=v_0$), \eqref{r.est} and \eqref{BA:shortpulse},
$$\f{1}{42}(v_0-u) -O(\de) \leq (-r\rd_u r)(u,v)\leq \f 23(v_0-u)+O(\de).$$
Using \eqref{r.est} again, and taking $\de$ sufficiently small depending on $v_0$, $u_1$, we have
\begin{equation}\label{dur.shortpulse}
\f 1{84}\leq (-\rd_u r)(u,v)\leq  \f{28}{3}.
\end{equation}

\pfstep{Step~2(c): Estimates for $r\rd_u\phi$ (Proof of \eqref{duphi.shortpulse})} Using $\rd_v (r\rd_u\phi)=-(\rd_u r)(\rd_v\phi)$, \eqref{r.est}, \eqref{BA:shortpulse} and \eqref{dur.shortpulse}, we obtain \eqref{duphi.shortpulse}.

In particular, by \eqref{eq:main:geo} and \eqref{eq:main:apriori} (for $r\rd_u\phi$ on $\{v=v_0\}$) and \eqref{r.est}, we have
\begin{equation}\label{duphi.2.shortpulse}
|\rd_u\phi|(u,v)\ls_{M_f, \,M_f^{-1}} \left(\eta+\de^{\f 12}\right).
\end{equation}

\pfstep{Step~2(d): Estimates for $r\rd_v\phi$ (Proof of \eqref{dvphi.shortpulse})} Using $\rd_u (r\rd_v\phi)=-(\rd_v r)(\rd_u\phi)$, \eqref{dvr.shortpulse} and \eqref{duphi.2.shortpulse}, we thus have \eqref{dvphi.shortpulse}.

In particular, by Lemma~\ref{lem:data}, we have the upper bound
\begin{equation}\label{dvphi.upper.bound}
r|\rd_v\phi|(u,v)\leq 3M_f^{\f 12}\de^{-\f 12},
\end{equation}
as well as the integrated lower bound
\begin{equation}\label{dvphi.lower.bound}
\int_{v_0}^{v_0+\de} r^2(\rd_v\phi)^2(u,v')\, \ud v'\geq \f 14 M_f,
\end{equation}
provided that $\dlt$ is sufficiently small (depending on $M_{f}$, $M_{f}^{-1}$ and $v_{0}$).

\pfstep{Step~2(e): Estimates for $\rd_v\log\Omg$ and $\Omg^2$} 
To control $\rd_v\log\Omg$, we first bound the $u$-integral of each term on the right hand side of the last equation in \eqref{WW.SS}. By \eqref{BA:shortpulse} (for $r |\rd_v\phi|$), \eqref{r.est} (for $r^{-1}$), \eqref{duphi.2.shortpulse} (for $|\rd_u\phi|$),
$$\int_{0}^{u_1} |\rd_u\phi||\rd_v\phi|(u,v) \,\ud u \ls_{M_f, \,M_f^{-1}} v_0 \de^{-\f 12} (\eta+\de^{\f 12})\ls_{M_f,\, M_f^{-1},\, v_0} \de^{-\f 12}.$$
Choosing $\de$ sufficiently small (depending on $M_f$, $M_f^{-1}$ and $v_0$), by \eqref{BA:shortpulse} (for $\Omg^2$), \eqref{r.est} (for $r^{-1}$), \eqref{dvr.shortpulse} (for $\rd_v r$) and \eqref{dur.shortpulse} (for $\rd_u r$), the remaining terms on the right hand side of \eqref{WW.SS} after integrating on $[0,u_1]$ are $O_{M_f,\, M_f^{-1},\, v_0} (\de^{-\f 12})$.
Combining the above estimates with \eqref{dvlogOmg.data.shortpulse} for the data, we have, for $\de$ sufficiently small, that
\begin{equation}\label{dvlogOmg}
|\rd_v\log\Omg|\leq 4 v_0^{-1}M_f \de^{-1}+ O_{M_f, \, M_f^{-1},\, v_0}(\de^{-\f 12}).
\end{equation}
By \eqref{eq:main:geo}, on $\{v=v_0\}$,
$$\f 29=4\cdot \f 13 \cdot \f 16<\Omg^2(u,v_0)=-\f{4\rd_v r \rd_u r}{1-\f{2m}{r}} <4\cdot\f 12\cdot \f 23\cdot \f{1}{2}=\f 23.$$
Using this and integrating \eqref{dvlogOmg}, we obtain
\begin{equation}\label{Omg.final.shortpulse}
\f 29(1-O(v_0^{-1}M_f)-O_{M_f, \, M_f^{-1},\, v_0} (\de^{\f 12})) \leq \Omg^2\leq \f 23(1+O(v_0^{-1}M_f)+O_{M_f, \, M_f^{-1},\, v_0} (\de^{\f 12})).
\end{equation}

\pfstep{Step~2(f): Completion of the bootstrap argument} 
By \eqref{dvphi.upper.bound} and \eqref{Omg.final.shortpulse}, we have improved the bootstrap assumptions \eqref{BA:shortpulse} as long as $v_{0}$ is large enough (depending on $M_{f}$) and $\de$ is sufficiently small.

\pfstep{Step~3: Improving the estimates (Proof of \eqref{phi.shortpulse}, \eqref{dvrOmg.bound} and \eqref{dvr.imp.bound})} First, using \eqref{eq:main:apriori} for $\phi$ and the estimate \eqref{dvphi.shortpulse}, we obtain \eqref{phi.shortpulse}. Next, integrating the first equation in \eqref{eqn.Ray}, using \eqref{eq:main:geo} (for $\f{\rd_v r}{\Omg^2}$ on $\{v=v_0\}$) and \eqref{dvphi.upper.bound}, we have
$$\left|\f{\rd_v r}{\Omg^2}\right|(u,v)\ls 1+\int_{v_0}^{v_0+\de} \f{1}{(v'-u)} M_f \de^{-1}\, \ud v'\ls 1,$$
which gives \eqref{dvrOmg.bound}. Finally, by \eqref{eq:SSESF}, for $u\in [0, u_0']$, 
$$\left|\log \f{\rd_v r(u,v)}{1/2}\right|=\int_{0}^u \f{2m(-\rd_u r)}{(1-\f{2m}{r})r^2}\, \ud u'
\ls_{M_i} 
\frac{1}{r(u_{0}', v)}
\ls
v_0^{-1}.$$
Exponentiating yields \eqref{dvr.imp.bound}.

\pfstep{Step~4: Formation of trapped surface} We now prove the remaining assertion regarding the formation of a trapped surface. First, we note that by \eqref{dur.shortpulse}, $(\rd_u r)(u_1, v_0+\de)<0$. It remains to show the negativity of $(\rd_v r)(u_1, v_0+\de)$.

By \eqref{eqn.Ray}, \eqref{eq:main:geo} (for $\f{\rd_v r}{\Omg^2}=\f{(1-\f{2m}{r})}{4(-\rd_u r)}$ on $\{v=v_0\}$), \eqref{r.est} (for $r^{-1}$), \eqref{BA:shortpulse} (for $\Omg^{-2}$) and \eqref{dvphi.lower.bound}, we have
\begin{equation*}
\begin{split}
\f{\rd_v r}{\Omg^2}(u_1,v_0+\de)=&\f{\rd_v r}{\Omg^2}(u_1,v_0)-\int_{v_0}^{v_0+\de} \f{r}{\Omg^2}(\rd_v\phi)^2(u_1,v')\, \ud v'\\
\leq & \frac{1 - \frac{2m}{r}}{4 (-\rd_{u} r)}(u_{1}, v_{0}) -\left(\inf_{v' \in [v_{0}, v_{0}+\dlt]} r^{-1} \Omg^{-2}(u_{1}, v')\right) \int_{v_{0}}^{v_{0}+\dlt} r^{2} (\rd_{v} \phi)^{2}(u_{1}, v') \, \ud v' \\
< &\f 32 - \f{3}{64}\f{1}{v_0+\de-u_1} M_f .
\end{split}
\end{equation*}

Since $\Omg^2>0$, and $u_1=v_0-\f{M_f}{64}$ (and hence $64(v_0+\de-u_1)<2M_f$ for sufficiently small $\de$), this yields $(\rd_v r)(u_1, v_0+\de)<0$. Hence, the sphere given by $(u,v)=(u_1, v_0+\de)$ is trapped.
\end{proof}

We now apply Proposition~\ref{prop:shortpulse} to compute the mass on the constant $v=v_0+\de$ hypersurface when $u\in \left[u_0', u_1\right]$:

\begin{proposition}\label{mass.shortpulse}
In the set $\left\{(u,v): u\in [u_0', u_1],\, v=v_0+\de \right\}$,
$$\left| m(u,v) -M_f\right|\ls_{M_i,\, M_f,\, M_f^{-1}} \left(v_0^{-1}+\eta^2\right). $$
\end{proposition}
\begin{proof}
The precise estimates for $r\rd_v\phi$ in Proposition~\ref{prop:shortpulse} allow us to show in Step~1 below that the mass difference on $u=u_0'$ for $v\in [v_0, v_0+\de]$ is, up to small error terms, $M_f$. Then in Step~2, we will bound the mass variation on the $v=v_0+\de$ hypersurface for $u\in \left[u_0', u_1\right]$

\pfstep{Step~1: Estimating the mass variation on $u=u_0'$, $v\in [v_0, v_0+\de]$} The key point here is that for appropriately chosen parameters, this hypersurface lies within a large $r$ region. Let $v\in [v_0,v_0+\de]$. By \eqref{dvr.imp.bound}, (and the obvious bound $\f{2m}{r}(u_0',v)\ls_{M_i} v_0^{-1}$), we have $|\f{\rd_v r}{1-\f{2m}{r}}(u_0',v)-\f 12|\ls_{M_i} v_0^{-1}$. Therefore, by \eqref{eq:SSESF} and \eqref{dvphi.upper.bound},
\begin{equation*}
\begin{split}
&\left|m(u_0',v_0+\de)-m(u_0',v_0)-\int_{v_0}^{v_0+\de} r^2(\rd_v\phi)^2(u_0',v')\, \ud v'\right|\\
\ls &\int_{v_0}^{v_0+\de} \left|\f{1-\f{2m}{r}}{\rd_v r}-2\right|r^2(\rd_v\phi)^2(u_0',v')\, \ud v'\ls_{M_i} v_0^{-1}.
\end{split}
\end{equation*}
We estimate the last term on the LHS as follows (using \eqref{chi.prop}):
\begin{equation*}
\begin{split}
&\left|\int_{v_0}^{v_0+\de} r^2(\rd_v\phi)^2(u_0',v')\, \ud v'-M_f \right|\\
\ls &\underbrace{\left|\int_{v_0}^{v_0+\de} r^2(\rd_v\phi)^2(u_0',v')\, \ud v'- \int_{v_0}^{v_0+\de} r^2(\rd_v\phi)^2(0,v')\, \ud v'\right|}_{=:I} \\
& + \underbrace{\left|\int_{v_0}^{v_0+\de} (\rd_v(r\phi))^2(0,v')\, \ud v'-M_f \int_0^1 \chi^2(s) \, \ud s\right|}_{=:II}\\
& + \underbrace{\left|\int_{v_0}^{v_0+\de} (\rd_v(r\phi))^2(0,v')\, \ud v'-\int_{v_0}^{v_0+\de} r^2(\rd_v\phi)^2(0,v')\, \ud v'\right|}_{=:III}.\\
\end{split}
\end{equation*}
By \eqref{dvphi.shortpulse} and \eqref{dvphi.upper.bound}, $I \ls_{M_f,\, M_f^{-1}} \de v_0^2(\eta+\de^{\f 12}) \dlt^{-\frac{1}{2}} \ls_{M_{f}, M_{f}^{-1}, v_{0}} \dlt^{\frac{1}{2}}.$
By \eqref{chi.prop} and Definition~\ref{def:data}, $II=0$.
Finally, by \eqref{phi.shortpulse}, \eqref{dvr.imp.bound} and \eqref{dvphi.upper.bound},
$$III\ls \left|\int_{v_0}^{v_0+\de}\left(r|\rd_v r||\phi||\rd_v\phi|+|\rd_v r|^2|\phi|^2\right)(0,v')\, \ud v'\right|\ls_{M_f, \, M_f^{-1},\, v_0} \de^{\f 12}.$$
Hence, using also Proposition~\ref{small.mass.first.region}, we conclude that
$$|m(u_0',v_0+\de)-M_f|\ls_{M_i} \left(v_0^{-1}+\eta^2\right)+O_{M_f,\, M_f^{-1},\, v_0}(\de^{\f 12})\ls_{M_i} \left(v_0^{-1}+\eta^2\right),$$
for $\de$ sufficiently small.

\pfstep{Step~2: Estimating the mass variation on $v=v_0+\de$ for $u\geq u_0'$} For this we need to use Proposition~\ref{small.mass.first.region}, which implies (since $m=0$ at $u=v=v_0$) that
\begin{equation}\label{small.mass.incoming}
\int_{u_0'}^{v_0} (-\rd_u m)(u, v_0)\, \ud u\ls_{M_i} \eta^2.
\end{equation}
Moreover, \eqref{eq:main:geo} together with \eqref{r.est} and \eqref{dvrOmg.bound} imply that for $u\in [0, u_1]$,
$$(-\rd_u m)(u,v_0)\gtrsim r^2(\rd_u\phi)^2(u, v_0).$$
Combining this with \eqref{small.mass.incoming} yields (since $u_1<v_0$)
$$\int_{u_0'}^{u_1} r^2(\rd_u\phi)^2(u,v_0) \, \ud u\ls_{M_i} \eta^2.$$
By \eqref{duphi.shortpulse}, and choosing $\de$ sufficiently small, it also holds that
$$\int_{u_0'}^{u_1} r^2(\rd_u\phi)^2(u,v_0+\de) \, \ud u\ls_{M_i,\,M_f,\, M_f^{-1}} \eta^2.$$
Hence, using \eqref{dvrOmg.bound}, for every $u\in [u_0', u_1]$,
\begin{equation*}
\begin{split}
\left|m(u,v_0+\de)-m(u_0',v_0+\de) \right|
\ls &\int_{u_0'}^{u_1} r^2(\rd_u\phi)^2(u', v_0+\de)\, \ud u'
\ls_{M_i,\,M_f,\, M_f^{-1}} \eta^2.
\end{split}
\end{equation*}
We conclude the proof of the Proposition by combining the conclusions of Steps~1 and 2.
\end{proof}

\subsection{The region after the short pulse}

\begin{proposition}\label{prop.afterpulse}
In the set $\{(u,v): u\in [0, u_0'],\, v\in[v_0+\de,\infty)\}$,
$$\left|\f{r\rd_v\phi}{\rd_v r}\right|(u,v) \ls_{M_i} \eta v_0^{\f 12} r^{-1}(u,v).$$
In particular, this holds on $\{(u_0',v): v\in [v_0+\de,\infty) \}$, (i.e., the constant $u=u_0'$ hypersurface restricted to the future of the short pulse region).
\end{proposition}
\begin{proof}
The key point is to exploit the decay in $r$ here. For this purpose, it will be convenient to consider the variables $r\phi$ and $\f{\rd_v(r\phi)}{\rd_v r}$. Note that $r\phi$ is \underline{not} small. Nevertheless, the quantity of interest $\f{r\rd_v\phi}{\rd_v r}$ is small.

\pfstep{Step~1: Bootstrap argument}
\textbf{In this proof, we take $(u,v)\in \{(u,v): u\in [0, u_0'],\, v\in[v_0+\de,\infty)\}$.}
Assume as a bootstrap assumption
\begin{equation}\label{BA:phi.last.region}
r|\phi|(u,v)\leq \eta v_0.
\end{equation}
Note that by \eqref{phi.shortpulse}, for $u\in [0, u_0']$, $r|\phi|(u,v_0+\de) \ls \eta r^{\f 12} +O_{M_f\, M_f^{-1}}(\de^{\f 12}) \ls \eta v_0^{\f 12}$ for $\de$ sufficiently small. Hence, \eqref{BA:phi.last.region} holds initially on\footnote{That it also holds initially on $\{u=0\}$ is trivial in view of Lemma~\ref{lem:data}.} $\{v=v_0 +\de \}$ (for $v_0$ sufficiently large).

By \eqref{eq:SSESF},
\begin{equation} \label{eq:dUdvpsi}
	\rd_{u} \left(\frac{\rd_{v}  (r\phi)}{\rd_{v} r}  \right) 
	= - \frac{\rd_{u} \rd_{v} r}{\rd_{v} r} \frac{1}{\rd_{v} r} \rd_{v} (r\phi) + \frac{\rd_{u} \rd_{v} r}{r \rd_{v} r} r\phi
	= - \frac{2m \rd_{u} r }{(1-\frac{2m}{r}) r^{2}} \f{\rd_{v} (r\phi)}{\rd_v r} + \frac{2m \rd_{u} r }{(1-\frac{2m}{r}) r^{3}}  r\phi.
\end{equation}
Integrating \eqref{eq:dUdvpsi} using \eqref{BA:phi.last.region} together with the following observations
\begin{itemize}
\item the initial data for $\f{\rd_v(r\phi)}{\rd_v r}$ are trivial;
\item $\int_{0}^{u_0'} \frac{2m (-\rd_{u} r) }{(1-\frac{2m}{r}) r^{2}}(u,v)\, \ud u\leq  \int_{r(u_0',v)}^{\infty} \frac{4m_i (-\rd_{u} r) }{r^{2}}\, \ud r\ls M_i v_0^{-1}$;
\item $\int_{0}^{u} \frac{2m (-\rd_{u} r) }{(1-\frac{2m}{r}) r^{3}} r \abs{\phi}(u,v)\, \ud u \leq \eta v_{0} \int_{r(u,v)}^{\infty} \frac{4m_i (-\rd_{u} r) }{r^{3}}\, \ud r\ls \eta v_{0} M_i  r^{-2}(u,v)$;
\end{itemize}
we obtain
\begin{equation}\label{r2dvrphi}
r^2\left|\frac{\rd_{v}  (r\phi)}{\rd_{v} r}\right|(u,v)\ls M_i\cdot e^{M_i v_0^{-1}} \eta v_{0} \ls_{M_i} \eta v_{0}. 
\end{equation}
Using the fundamental theorem of calculus and \eqref{r2dvrphi}, we obtain
\begin{equation}\label{rphi.est.final}
\begin{split}
r|\phi|(u,v)\ls_{M_i}& \eta v_0^{\f 12}+\left(\sup_{v''\in [v_0+\de, v]}r^2\left|\frac{\rd_{v}  (r\phi)}{\rd_{v} r}\right|(u,v'')\right)\int_{v_0+\de}^v r^{-2} (\rd_v r)(u,v')\, \ud v' \\
\ls_{M_i}& \eta v_0^{\f 12}+ \eta v_{0} \int_{r(u,v_0+\de)}^\infty \f{\ud r'}{(r')^2}\ls \eta v_0^{\f 12}+ \eta \ls \eta v_0^{\f 12}.
\end{split}
\end{equation}
For $v_0$ sufficiently large, this improves \eqref{BA:phi.last.region} and we have thus closed the bootstrap argument.
\pfstep{Step~2: Estimates for $\rd_v\phi$} Finally, we compute using \eqref{r2dvrphi} and \eqref{rphi.est.final} that
\begin{equation*}
\left|\f{r\rd_v\phi}{\rd_v r}\right|(u,v) =\left|\f{\rd_v(r\phi)}{\rd_v r}-\phi\right|\ls_{M_i} \eta v_0^{\f 12} r^{-1}(u,v). \qedhere
\end{equation*}

\end{proof}

The previous estimate allows us to compute the mass on the hypersurface $\{u=u_0'\}$:
\begin{proposition}\label{mass.final}
For $v\in [v_0+\de,\infty)$,
$$\left|m \left(u_0',v\right)- M_f\right|\ls_{M_i} \eta^2+v_0^{-1} .$$
\end{proposition}
\begin{proof}
Integrating the last equation in \eqref{eq:SSESF}, for any $v\geq v_0+\de$,
\begin{equation*}
\begin{split}
m(u_0',v)=& e^{-\int_{v_0+\de}^v r\f{(\rd_v\phi)^2}{\rd_v r}(u_0',v') \, \ud v'} m(u_0', v_0+\de)\\
& +\f 12 e^{-\int_{v_0+\de}^v r\f{(\rd_v\phi)^2}{\rd_v r}(u_0',v') \, \ud v'} \int_{v_0+\de}^v  r^2\f{(\rd_v\phi)^2}{\rd_v r}(u_{0}',v') e^{\int_{v_0+\de}^{v'} r\f{(\rd_v\phi)^2}{\rd_v r}(u_0',v'') \, \ud v''}\, \ud v'.
\end{split}
\end{equation*}
We first show that the integrals in the exponentials are small. More precisely, by Proposition~\ref{prop.afterpulse},
\begin{equation*}
\begin{split}
\int_{v_0+\de}^v r\f{(\rd_v\phi)^2}{\rd_v r}(u_0',v') \, \ud v'\ls_{M_i} & \eta^{2} v_0 \int_{v_0+\de}^\infty \f{\rd_v r}{r^3}(u_0',v')\,\ud v'\\
\ls_{M_i} & \eta^{2} v_0 \int_{r(u_0',v_0+\de)}^\infty \f{\ud r'}{(r')^3} \ls_{M_i} \eta^{2} v_0^{-1}.
\end{split}
\end{equation*}
A similar argument using Proposition~\ref{prop.afterpulse} yields
\begin{equation*}
\begin{split}
\int_{v_0+\de}^v r^{2} \f{(\rd_v\phi)^2}{\rd_v r}(u_0',v') \, \ud v'\ls_{M_i} \eta^2.
\end{split}
\end{equation*}
Finally, by Proposition~\ref{mass.shortpulse}, $\left|m(u_0', v_0+\de)-M_f\right|\ls_{M_i} \eta^2 + v_{0}^{-1}$.
Combining these observations give the desired conclusion.
\end{proof}

\subsection{Conclusion of the proof using monotonicity of mass}

\begin{proposition}
The final mass $m_f$ satisfies $|m_f-M_f|\leq \ep$.
\end{proposition}

\begin{proof}
Choosing $v_*$ sufficiently large, $\eta_{0}$$(\geq \eta)$ sufficiently small, (which then fixes $v_0$), and then $\de$ sufficiently small, Propositions~\ref{mass.shortpulse} and \ref{mass.final} imply that $|m(u,v)-M_f|<\ep$ on $\underline{C}_{(far)} \cup C_{(far)} := \{(u_0',v):v\in [v_0+\de,\infty) \}\cup \{(u,v_0+\de):u\in [u_0',u_1] \}$. These two hypersurfaces are depicted in the Penrose diagram in Figure~\ref{fig:mass}. 
Notice that by Proposition~\ref{prop:shortpulse}, the $2$-sphere given by $(u,v)=(u_1,v_0+\de)$ is trapped. Therefore, \eqref{eqn.Ray} implies that this sphere is \underline{not} causally connected to future null infinity (see Figure~\ref{fig:mass}). By the monotonicity of $m$ (increasing in $v$ and decreasing in $u$, cf. \eqref{eq:SSESF}) in the exterior region (i.e., the past of future null infinity), we thus infer that in the exterior region, for $v\geq v_0+\de$ and $u\geq u_1$, $|m(u,v)-M_f|<\ep$. In particular, $m_f$, which is a limit first in $v$ and then in $u$, must obey the desired bound.
\end{proof}

\begin{figure}[h] 
\begin{center}
\def\svgwidth{200px}
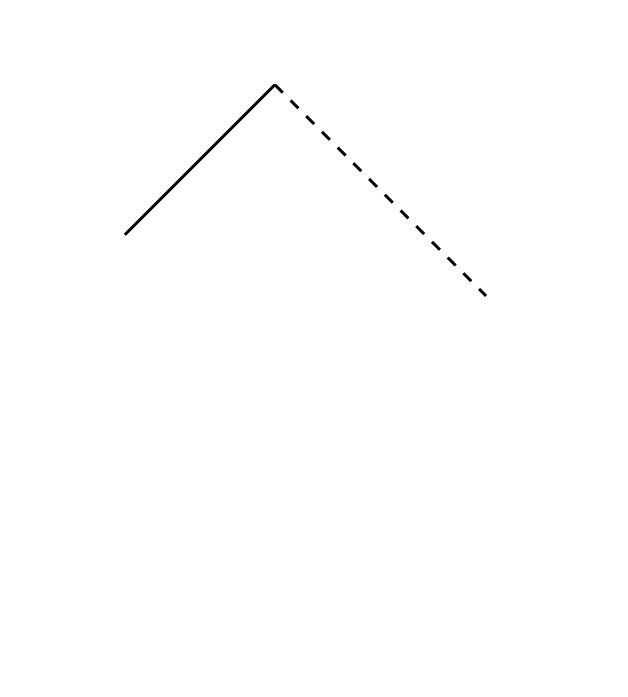 
\caption{} \label{fig:mass}
\end{center}
\end{figure}

\begin{remark}
Note that in the argument, we have not controlled the solution in the full exterior region of the black hole. Nevertheless, using the monotonicity of mass, the information on a subset of the exterior region is sufficient for computing the final mass up to an acceptable error.
\end{remark}

We summarize what we have achieved so far as follows:
\begin{theorem} \label{thm:main-eta}
Given $M_{i} \geq M_{f} > 0$ and $\eps > 0$, let $v_{\ast}^{-1}$, $\eta_{0}$ and $\dlt$ be sufficiently small positive numbers (chosen in this order, where the quantity on the right may depend on those on the left). Then for any $\eta \in [0, \eta_{0}]$, the solution $(r, \phi, m)$ to \eqref{ESS} with the initial data specified in Section~\ref{subsec:id} obeys
\begin{equation*}
	\left| m_{i} - \left(\frac{\eta}{\eta_{0}}\right)^{2} (M_{i} - M_{f}) + M_{f} \right| \leq \eps, \quad
	\abs{m_{f} - M_{f}} \leq \eps.
\end{equation*}

\end{theorem}
Choosing $\eta = \eta_{0}$, Theorem~\ref{thm:main} immediately follows. We will use the flexibility in the choice of $\eta$ in Section~\ref{subsec:no-eps} below, where we discuss a procedure for removing the $\eps$ from Theorem~\ref{thm:main}.

\subsection{Removing the $\eps$-error} \label{subsec:no-eps}

Here we elaborate on Remark~\ref{rem:no-eps} and give a proof of the strengthening of Theorem~\ref{thm:main} claimed there, assuming the following continuity property of $m_{f}$:
\begin{proposition} \label{prop:mf-cont-dep}
Let $(r_{(\eta)}, \phi_{(\eta)}, m_{(\eta)})$ be the solution to \eqref{ESS} given by Theorem~\ref{thm:main-eta}.
If $v_{\ast}^{-1}$, $\eta_{0}$ and $\dlt$ are sufficiently small, then the final Bondi mass 
\begin{equation*}
m_{f (\eta)} = \lim_{u \to u_{\calE \calH}^{-}} \lim_{v \to \infty} m_{(\eta)}(u, v)
\end{equation*}
depends continuously on $\eta \in [0, \eta_{0}]$.
\end{proposition}

As discussed in Remark~\ref{rem:no-eps}, Proposition~\ref{prop:mf-cont-dep} would follow from asymptotic stability of the exterior of the Schwarzschild solution, where the constants in the decay estimates are bounded in terms of a suitable norm of the initial data. To keep this note brief, we will refrain from giving a proof of Proposition~\ref{prop:mf-cont-dep} here. For the interested reader, we refer to \cite{DR, holzegel, LO-ext}, which contain relevant techniques to establish Proposition~\ref{prop:mf-cont-dep}.

If Proposition~\ref{prop:mf-cont-dep} were valid, then the statement given in Remark~\ref{rem:no-eps} could be established using Theorem~\ref{thm:main-eta} and a soft continuity argument.
\begin{proof}[Proof of Remark~\ref{rem:no-eps} assuming Proposition~\ref{prop:mf-cont-dep}]
In view of the scaling transformation $(r, \phi, m) \mapsto (a r, \phi, a m )$, it suffices to construct a solution whose final-to-initial-Bondi-mass ratio $\frac{m_{f}}{m_{i}}$ coincides with the given ratio $\frac{M_{f}}{M_{i}} < 1$. Applying Theorem~\ref{thm:main-eta} with a bigger $M_{i}$ (namely, $2 M_{i} - M_{f}$) but the same $M_{f}$, we obtain a one-parameter family $(r_{(\eta)}, \phi_{(\eta)}, m_{(\eta)})$ of solutions such that
\begin{gather*}
	m_{i (0)} = M_{f} + O(\eps), \quad m_{f(0)} = M_{f} + O(\eps), \\
	m_{i (\eta_{0})} = 2 M_{i} - M_{f} + O(\eps), \quad m_{f(\eta_{0})} = M_{f} + O(\eps).
\end{gather*}
In particular, choosing $\eps > 0$ small compared to $\frac{(M_{i} - M_{f}) M_{f}}{M_{i} + M_{f}}$, which is \emph{positive} by the hypothesis $M_{i} > M_{f}$, it can be arranged so that
\begin{equation*}
	\frac{m_{f (\eta_{0})}}{m_{i (\eta_{0})}} < \frac{M_{f}}{M_{i}} < \frac{m_{f (0)}}{m_{i (0)}}.
\end{equation*}

By Proposition~\ref{prop:mf-cont-dep} (as well as the continuity of $\eta \mapsto m_{i(\eta)}$, which is trivial), it follows that $\eta \mapsto \frac{m_{f (\eta)}}{m_{i (\eta)}}$ is continuous. Therefore, by the intermediate value theorem, the desired solution can be found among $(r_{(\eta)}, \phi_{(\eta)}, m_{(\eta)})_{\eta \in [0, \eta_{0}]}$. \qedhere
\end{proof}

\end{document}